\documentclass[11pt,letterpaper]{article}
\usepackage[margin=1in]{geometry}
\usepackage[utf8x]{inputenc}
\usepackage{amsmath, amsthm}
\usepackage{wrapfig,graphicx,amssymb,textcomp,array,amsmath}
\usepackage{algpseudocode} 
\usepackage{enumerate}
\usepackage{enumitem}
\usepackage{multirow}
\usepackage{tabularx}
\algtext*{EndWhile}
\algtext*{EndIf}
\usepackage{algorithm}
\usepackage{color}
\usepackage[titletoc,title]{appendix}

\newcommand{\linepq}[2]{\overline{#1#2}}

\newcommand{\cpd}{\mu}

\newcommand{\changed}[1]{{\color{black} #1}}

\title{Flip Distance to some Plane Configurations}

\author{Ahmad Biniaz\thanks{Cheriton School of Computer Science, University of Waterloo. Supported by NSERC and Fields Postdoctoral Fellowships. ahmad.biniaz@gmail.com}
	\and Anil Maheshwari\thanks{School of Computer Science, Carleton University. Supported by NSERC. \{anil, michiel\}@scs.carleton.ca} 
	\and  Michiel Smid\footnotemark[2]}

\newtheorem{lemma}{Lemma}

\newtheorem{theorem}{Theorem}

\newtheorem*{problem*}{Problem}

\begin{document}
\maketitle
\begin{abstract}
We study an old geometric optimization problem in the plane. Given a perfect matching $M$ on a set of $n$ points in the plane, we can transform it to a non-crossing perfect matching by a finite sequence of flip operations. The flip operation removes two crossing edges from $M$ and adds two non-crossing edges. Let $f(M)$ and $F(M)$ denote the minimum and maximum lengths of a flip sequence on $M$, respectively. It has been proved by Bonnet and Miltzow (2016) that $f(M)=O(n^2)$ and by van Leeuwen and Schoone (1980) that $F(M)=O(n^3)$. 
We prove that $f(M)=O(n\Delta)$ where $\Delta$ is the spread of the point set, which is defined as the ratio between the longest and the shortest pairwise distances. This improves the previous bound \changed{if the point set has sublinear spread.}
For a matching $M$ on $n$ points in convex position we prove that $f(M)=n/2-1$ and $F(M)={{n/2} \choose 2}$; these bounds are tight.

Any bound on $F(\cdot)$ carries over to the bichromatic setting, while this is not necessarily true for $f(\cdot)$. Let $M'$ be a bichromatic matching. The best known upper bound for $f(M')$ is the same as for $F(M')$, which is essentially $O(n^3)$. We prove that $f(M')\leqslant n-2$ for points in convex position, and $f(M')= O(n^2)$ for semi-collinear points.

The flip operation can also be defined on spanning trees. 
For a spanning tree $T$ on a convex point set we show that $f(T)=O(n\log n)$.
\end{abstract}
\section{Introduction}
\label{introduction-section}
A {\em geometric graph} is a graph whose vertices are points in the plane, and whose edges are straight-line segments connecting the points. All graphs that we consider in this paper are geometric. A graph is {\em plane} if no pair of its edges cross each other. Let $n\geqslant 2$ be an even integer, and let $P$ be a set of $n$ points in the plane that is in general position (no three points on a line). For two points $a$ and $b$ in the plane, we denote by $ab$ the segment with endpoints $a$ and $b$. Let $M$ be a perfect matching on $P$. If two edges in $M$ cross each other, we can remove this crossing by a flip operation. The {\em flip operation} (or flip for short) removes two crossing edges and adds two non-crossing edges to obtain a new perfect matching. In other words, if two segments $ab$ and $cd$ cross, then a flip removes $ab$ and $cd$ from the matching, and adds either $ac$ and $bd$, or $ad$ and $bc$ to the matching; see Figure~\ref{flip-fig}(a). Every flip decreases the total length of the edges of $M$, and thus, after a finite sequence of flips, $M$ can be transformed to a plane perfect matching. This process of transforming a crossing matching to a plane matching is referred to as {\em uncrossing} or {\em untangling} a matching. Motivated by this old folklore result, we investigate the minimum and the maximum lengths of a sequence of flips to reach a plane matching.

To uncross a perfect matching $M$, we say that the sequence $(M{=}M_0,M_1,\dots, M_k)$ is a valid flip sequence if $M_{i+1}$ is obtained from $M_i$ by a single flip, and $M_k$ is plane. The number $k$ denotes the length of this flip sequence. We define $f(M)$ to be the minimum length of any valid flip sequence to uncross $M$, that is, the minimum number of flips required to transform $M$ to a plane perfect matching. We define $F(M)$ to be the maximum length of any valid flip sequence. As for $F(M)$, one can imagine that an adversary imposes which of the two flips to apply on which of the crossings.

In the bichromatic setting, we are given $n/2$ red and $n/2$ blue points and a bichromatic matching, that is a perfect matching in which the two endpoints of every segment have distinct colors. Contrary to the monochromatic setting, there is only one way to flip two crossing bichromatic edges; see Figure~\ref{flip-fig}(b). In the bichromatic setting the adversary can only impose the crossing to flip. Thus, any upper bound on $F(M)$ for monochromatic matchings carries over to bichromatic matchings; this statement is not necessarily true for $f(M)$.

\begin{figure}[htb]
	\centering
	\setlength{\tabcolsep}{0in}
	$\begin{tabular}{cc}
	\multicolumn{1}{m{.6\columnwidth}}{\centering\includegraphics[width=.45\columnwidth]{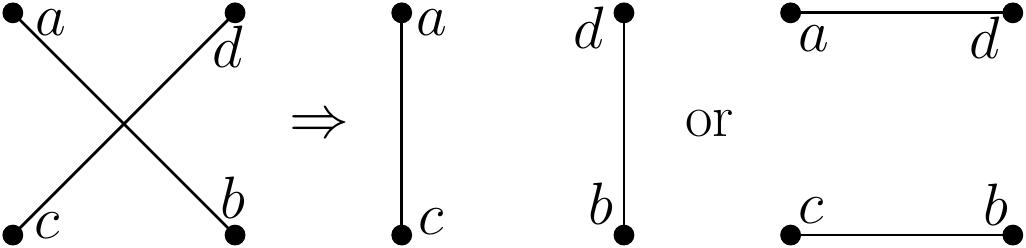}}
	&\multicolumn{1}{m{.4\columnwidth}}{\centering\includegraphics[width=.27\columnwidth]{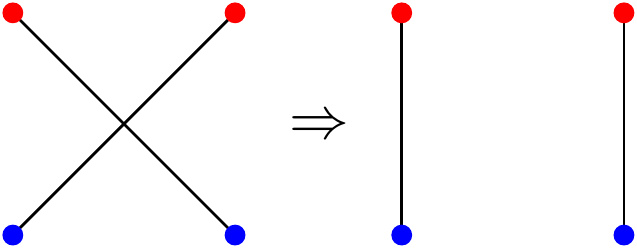}}
	\\
	(a)&(b)
	\\
	\multicolumn{1}{m{.6\columnwidth}}{\centering\includegraphics[width=.45\columnwidth]{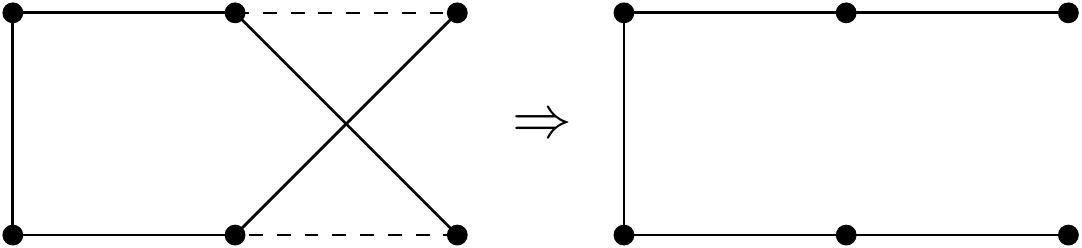}}
	&\multicolumn{1}{m{.4\columnwidth}}{\centering\includegraphics[width=.25\columnwidth]{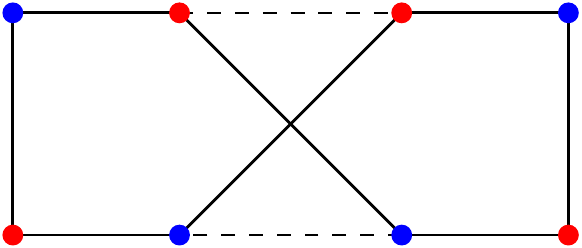}}
	\\
	(c) & (d) 
	\end{tabular}$
	\caption{(a) Two ways to flip a crossing in a monochromatic matching. (b) The only way to flip a crossing in a bichromatic matching. (c) One way to flip a crossing in a monochromatic tree. (d) No way to flip a crossing in a bichromatic Hamiltonian cycle.}
	\label{flip-fig}
\end{figure}

The flip operation can be defined for a spanning tree (resp.~a Hamiltonian cycle) analogously, that is, we remove a pair of crossing edges and add two other edges so that the graph remains a spanning tree (resp.~a Hamiltonian cycle) after this operation. We define $f(\cdot)$ and $F(\cdot)$ for spanning trees and Hamiltonian cycles, analogously. As shown in Figure~\ref{flip-fig}(c), there is only one way to flip a crossing in a spanning tree (resp.~a Hamiltonian cycle). Contrary to the bichromatic matching, it is not always possible to flip a crossing in a bichromatic spanning tree nor in a bichromatic Hamiltonian cycle; see Figure~\ref{flip-fig}(d).

\subsection{Related Work}
The most relevant works are by van Leeuwen and Schoone \cite{Leeuwen1981}, and Oda and Watanabe \cite{Oda2007} for Hamiltonian cycles, and by Bonnet and Miltzow \cite{Bonnet2016} for matchings. They proved, with elegant arguments, the following results.

\begin{theorem}[van Leeuwen and Schoone, 1981 \cite{Leeuwen1981}]
	\label{Leeuwen-thr}
	For every Hamiltonian cycle $H$ on $n$ points in the plane we have that $F(H)=O(n^3)$.
\end{theorem}

\begin{theorem}[Oda and Watanabe, 2007 \cite{Oda2007}]
	\label{Oda-thr}
	For every Hamiltonian cycle $H$ on $n$ points in the plane in convex position we have that $f(H)\leqslant 2n-7$.
\end{theorem}

As for a lower bound, they presented a Hamiltonian cycle $H$ on $n\geqslant 7$ points in the plane in convex position for which $f(H)\geqslant n-2$. 

\begin{theorem}[Bonnet and Miltzow, 2016 \cite{Bonnet2016}]
	\label{Bonnet-thr}
	For every perfect matching $M$ on a set of $n$ points in the plane in general position we have that $f(M)=O(n^2)$. 
\end{theorem}

The $O(n^3)$ upper bound of Theorem~\ref{Leeuwen-thr} carries over to perfect matchings. As for lower bounds, Bonnet and Miltzow \cite{Bonnet2016} presented two matchings $M_1$ and $M_2$ such that $f(M_1)=\Omega(n)$ and $F(M_2)=\Omega(n^2)$. 
The bound $F(M)=O(n^3)$ holds even if $M$ is a bichromatic matching, while the proof of $f(M)=O(n^2)$ does not generalize for the bichromatic setting. 

An alternate definition of an {\em edge flip} in a graph is the operation of removing one edge and inserting a different edge such that the resulting graph remains in the same graph class. \changed{The problem of minimizing the number of edge flips, to transform a graph into another, has been studied for many different graph classes.} See the survey by Bose and
Hurtado \cite{Bose2009} on edge flips in planar graphs both in the combinatorial and the geometric settings, and see \cite{Aichholzer2015b, Bose2011, Hurtado1999, Lawson1972} for edge flips in triangulations.

A related problem is the compatible matching problem in which we are given two perfect matchings on the same point set and the goal is to transform one to another by a sequence of compatible matchings (two perfect matchings, on the same point set, are said to be compatible if they are edge disjoint and their union is non-crossing). See \cite{Aichholzer2015, Aichholzer2009, Aloupis2015, Ishaque2013} for recent work on compatible matchings, and \cite{Olaverri2014} for its extension to compatible trees.

\subsection{Our Contribution}
In this paper we decrease the gap between lower and upper bounds for $f(\cdot)$ and $F(\cdot)$ for some input configurations. 
In Section~\ref{dense} we show that for every perfect matching $M$, on a set $P$ of $n$ points in the plane, we have $f(M)=O(n\Delta)$ where $\Delta$ is the spread of $P$.

Assume that $P$ is in convex position. In Section~\ref{convex-section} we show that for every perfect matching $M$ on $P$ we have that $f(M)\leqslant n/2-1$ and $F(M)\leqslant {n/2 \choose 2}$. These bounds are tight as Bonnet and Miltzow \cite{Bonnet2016} showed the existence of two perfect matchings $M_1$ and $M_2$ on $n$ points in convex position such that $f(M_1)\geqslant n/2-1$ and $F(M_2)\geqslant {n/2\choose 2}$. We also prove that for every spanning tree $T$ on $P$ we have that $f(T)=O(n\log n)$.

In Section~\ref{bichromatic-section} we study bichromatic matchings on special point sets. Assume that the points of $P$ are colored red and blue. We prove that, if $P$ is in convex position, then for every perfect bichromatic matching $M$ on $P$ we have that $f(M)\leqslant n-2$. Also, we prove that, if $P$ is semi-collinear, i.e., the blue points are on a straight line, then for every perfect bichromatic matching $M$ on $P$ we have that $f(M)=O(n^2)$. Table~\ref{table1} summarizes the results.

\begin{table}
	\centering
	\caption{Upper bounds on the minimum and maximum number of flips ($\Delta$ is the spread).}
	\label{table1}
	\begin{minipage}{14cm}\centering
		\begin{tabular}{|l|cc|cc|}
			\cline{1-5}
			{\bf minimum \# of flips}& \multicolumn{2}{c|}{$f(\cdot)$-general position}&\multicolumn{2}{c|}{$f(\cdot)$-convex position} \\ 
			
			\cline{1-5}
			\multirow{2}{*}{matchings}& $O(n^2)$ & \cite{Bonnet2016}&\multirow{2}{*}{$n/2-1$} &\multirow{2}{*}{Theorem~\ref{matching-convex-thr}}\\
			& $O(n\Delta)$ & Theorem~\ref{dense-thr}&&\\

			\cline{1-5}
			bichromatic matchings& $O(n^3)$&\cite{Leeuwen1981}&$n-2$&Theorem~\ref{bimatching-convex-thr}\\ 
			
			\cline{1-5}
			trees&$O(n^3)$&\cite{Leeuwen1981}&$O(n\log n)$&Theorem~\ref{tree-thr}\\ 
			
			\cline{1-5}
			Hamiltonian cycles&$O(n^3)$&\cite{Leeuwen1981}&$2n-7$&\cite{Oda2007}\\
			
			\cline{1-5}
			\multicolumn{5}{c}{\vspace{-10pt}}\\ 
			
			\cline{1-5}
			\multicolumn{4}{|l}{bichromatic matching on semi-collinear points $f(\cdot) = O(n^2)$}  &Theorem~\ref{semi-collinear-thr}\\ 
			
			\cline{1-5}
			\multicolumn{5}{c}{\vspace{-2pt}}\\
			
			\cline{1-5}
			{\bf maximum \# of flips}& \multicolumn{2}{c|}{$F(\cdot)$-general position}&\multicolumn{2}{c|}{$F(\cdot)$-convex position} \\ 
			
			\cline{1-5}
			matchings/trees/cycles& $O(n^3)$&\cite{Leeuwen1981}&${n/2}\choose 2$&Theorem~\ref{matching-convex-thr}\\ 
			\cline{1-5}	
			
		\end{tabular}
	\end{minipage}
	
\end{table}

\subsection{Preliminaries}
\label{preliminaries-section}
Let $a$ and $b$ be two points in the plane. We denote by $ab$ the straight line-segment between $a$ and $b$, and by $\linepq{a}{b}$ the line through $a$ and $b$. Let $P$ be a set of points in the plane in convex position. For two points $p$ and $q$ in $P$ we define the {\em depth} of the segment $pq$ as the minimum number of points of $P\setminus\{p,q\}$ on either side of $\linepq{p}{q}$. A {\em boundary edge} is a segment of depth zero, i.e., an edge of the convex hull of $P$. An edge $e$ in a graph $G$ is said to be {\em free} if $e$ is not crossed by other edges of $G$.

\section{Minimum Number of Flips}
\label{dense}
The {\em spread} $\Delta$ of a set of points (also called the {\em distance ratio} \cite{Clarkson1999}) is the ratio between the largest and the smallest interpoint distances.
It is well known that the spread of a set of $n$ points in the plane is $\Omega(\sqrt{n})$ (see e.g., \cite{Valtr1994}).
In this section, we prove an upper bound on the minimum length of a flip sequence in terms of $n$ and $\Delta$. In fact we prove the following theorem.

\begin{theorem}
	\label{dense-thr}
	For every perfect matching $M$ on a set of $n$ points in the plane in general position we have that $f(M)=O(n\Delta)$, where $\Delta$ is the spread of the point set. 
\end{theorem}

For point sets with spread $o(n)$, the upper bound of Theorem~\ref{dense-thr} is better than the $O(n^2)$ upper bound of Theorem~\ref{Bonnet-thr}. For example, for {\em dense} point sets, which have spread $O(\sqrt{n})$, Theorem~\ref{dense-thr} gives an upper bound of $O(n\sqrt{n})$ on the number of flips. According to \cite{Edelsbrunner1997}, dense point sets commonly appear in nature, and they have applications in computer graphics. Valtr and others~\cite{Edelsbrunner1997, Valtr1994, Valtr1996} have established several combinatorial bounds for dense point sets that improve corresponding bounds for arbitrary point sets.

Let $P$ be a set of $n$ points in the plane with spread $\Delta$. Let $M$ be a perfect matching on $P$. We prove that $M$ can be untangled by $O(n\Delta)$ flips, i.e., $f(M)=O(n\Delta)$. The main idea of our proof is as follows.
Let $\cpd$ be the minimum distance between any pair of points in $P$. Let $|pq|$ denote the Euclidean distance between two points $p,q\in P$. Since $P$ has spread $\Delta$, we have
$|pq|\leqslant \cpd \Delta.$  
For the matching $M$ we define its {\em weight}, $w(M)$, to be the total length of its edges. Since $M$ has $n/2$ edges,
\begin{equation}
\label{matching-weight}
w(M)=\sum_{pq\in M} |pq|=O(n\cpd\Delta).
\end{equation}

Recall that a pair of crossing segments can be flipped in two different ways as depicted in Figure~\ref{flip-fig}(a). In the remainder of this section we show that one of these two flip operations reduces $w(M)$ by at least $t\cpd$, for some constant $t>0$. Combining this with Equality~\eqref{matching-weight} implies the existence of a flip sequence of length $O(n\Delta)$ that uncrosses $M$. 

Take any two crossing edges $pq$ and $rs$ in $M$, and let $o$ be their intersection point. We flip $pq$ and $rs$ to $ps$ and $rq$, if $\angle roq\leqslant \pi/2$, and to $pr$ and $qs$, otherwise. In other words, we flip $pq$ and $rs$ to the two edges that face the two smaller angles at $o$. In Lemma~\ref{flip-lemma} we prove that this flip reduces the length of edges by at least \changed{$\frac{2-\sqrt{2}}{4}\cpd'$}, where $\cpd'$ is the minimum distance between any pair of points in $\{p,q,r,s\}$. Since the minimum distance between pairs in $\{p,q,r,s\}$ is at least the minimum distance between pairs in $P$, our result follows.  
We use the following two lemmas in the proof of Lemma~\ref{flip-lemma}; we prove these two lemmas later.

\begin{lemma}
	\label{flip-lemma1}
	Let $ab$ and $cd$ be two crossing segments, and let $o$ be their intersection point. Let $\cpd''$ be the minimum distance between any pair of points in $\{a, b, c, d\}$. If $\angle cob\leqslant \pi/3$, then $$(|ab|+|cd|)-(|ad|+|cb|)\geqslant \cpd''.$$
\end{lemma}

\begin{lemma}
	\label{flip-lemma2}
	Let $ab$ and $cd$ be two perpendicular segments that cross each other. Let $\cpd''$ be the minimum distance between any pair of points in $\{a, b, c, d\}$. Then, $$(|ab|+|cd|)-(|ad|+|cb|)\geqslant \changed{\frac{2-\sqrt{2}}{2}} \cpd''.$$
\end{lemma}

\begin{figure}[htb]
	\centering
	\setlength{\tabcolsep}{0in}
	$\begin{tabular}{ccc}
	\multicolumn{1}{m{.33\columnwidth}}{\centering\includegraphics[width=.3\columnwidth]{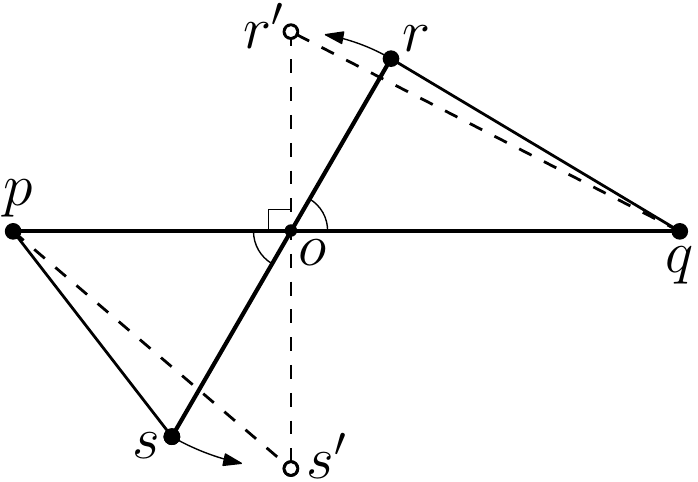}}
	&
	\multicolumn{1}{m{.33\columnwidth}}{\centering\includegraphics[width=.3\columnwidth]{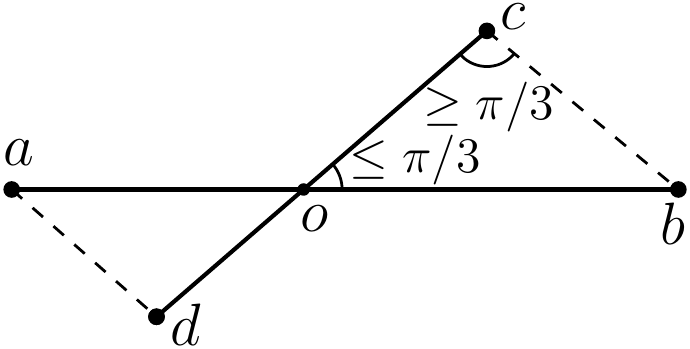}}
	&\multicolumn{1}{m{.33\columnwidth}}{\centering\includegraphics[width=.28\columnwidth]{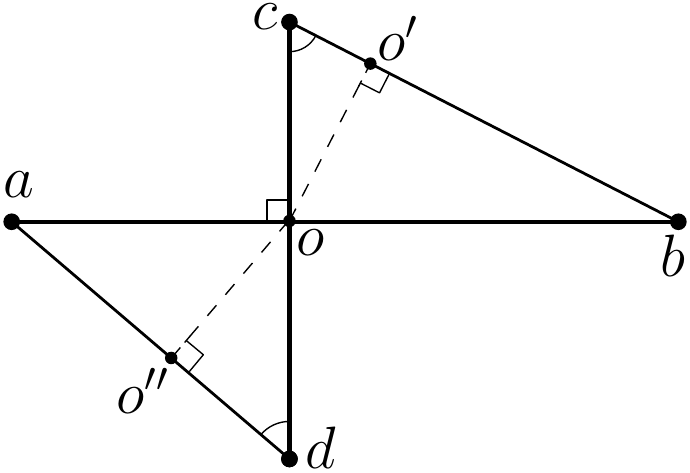}}
	\\
	(a) & (b)& (c)
	\end{tabular}$
	\caption{Illustrations of the proofs of (a) Lemma~\ref{flip-lemma}, (b) Lemma~\ref{flip-lemma1}, and (c) Lemma~\ref{flip-lemma2}.}
	\label{dense-fig}
\end{figure}

\begin{lemma}
	\label{flip-lemma}
	Let $pq$ and $rs$ be two crossing segments, and let $o$ be their intersection point. Let $\cpd'$ be the minimum distance between any pair of points in $\{p, q, r, s\}$. If $\angle roq\leqslant \pi/2$, then $$(|pq|+|rs|)-(|ps|+|rq|)\geqslant \changed{\frac{2-\sqrt{2}}{4}} \cpd'.$$
\end{lemma}

\begin{proof}
	If $\angle roq< \pi/3$, then our claim follows from Lemma~\ref{flip-lemma1} where $p,q,r,s$ play the roles of $a,b,c,d$, respectively. Assume that $\angle roq\geqslant \pi/3$. Observe that $\angle roq =\angle pos$. 
	
	After a suitable rotation and/or a horizontal reflection and/or relabeling assume that $|pq|\geqslant |rs|$, $pq$ is horizontal, $p$ is to the left of $q$, and $r$ lies above $pq$. Rotate $rs$ counterclockwise about $o$, while keeping $o$ on this segment, until $rs$ is vertical. See Figure~\ref{dense-fig}(a). After this rotation, let $r'$ and $s'$ denote the two points that correspond to $r$ and $s$, respectively. 
	
	\vspace{7pt}	
	\noindent{\em Claim $1$. $|r'p|> |rp|/2$ and $|qs'|> |qs|/2$.}
	\vspace{7pt}		
	
	We prove only the first inequality of this claim; the proof of the second inequality is analogous. Since $r'p$ is the hypotenuse of the right triangle $\bigtriangleup r'op$, we have $|r'o|\leqslant |r'p|$. Since $\bigtriangleup r'or$ is isosceles and $\angle r'or\leqslant \pi/6$, we have $|rr'|<|r'o|$, and thus, $|rr'|<|r'p|$. By the triangle inequality we have $|rp|\leqslant |rr'|+|r'p|<2|r'p|$, which implies $|r'p|> |rp|/2$. This proves Claim 1. 
	\vspace{7pt}
	
	Observe that $|r'q|\geqslant |rq|$, $|ps'|\geqslant |ps|$, $|r's'|=|rs|$, and by Claim 1, $|r'p|\geqslant |rp|/2$ and $|qs'|\geqslant |qs|/2$. Thus, the minimum distance $\cpd''$ between any pair of points in $\{p,q,r',s'\}$ is not smaller than half the minimum distance between any pair of points in $\{p,q,r,s\}$, i.e., $\cpd''\geqslant \cpd'/2$. \changed{By Lemma~\ref{flip-lemma2} we have $(|pq|+|r's'|)-(|ps'|+|r'q|)\geqslant \frac{2-\sqrt{2}}{2}\cpd''$}, where $p,q,r',s'$ play the roles of $a,b,c,d$, respectively. Combining these inequalities, we get 
	\begin{align}
	\notag	(|pq|+|rs|)-(|ps|+|rq|)	& \geqslant (|pq|+|r's'|)-(|ps'|+|r'q|)\\
	\notag							& \geqslant \frac{2-\sqrt{2}}{2}\cpd'' \geqslant \frac{2-\sqrt{2}}{4}\cpd'.
	\end{align}
\end{proof}

\noindent{\bf Note 1.}{ The constants $\frac{2-\sqrt{2}}{2}$ and $\frac{2-\sqrt{2}}{4}$ in the proofs of Lemmas~\ref{flip-lemma2} and \ref{flip-lemma} are not optimized. To keep our proofs short and simple, we avoid optimizing these constants.}

\vspace{5pt}
\begin{wrapfigure}{r}{0.32\textwidth}
	\begin{center}
		\vspace{-7pt}
		\includegraphics[width=.3\textwidth]{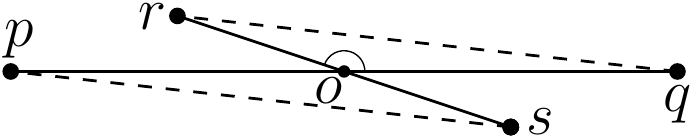}
	\end{center}
	\vspace{-12pt}
\end{wrapfigure}
\noindent{\bf Note 2.}{ The angle constraint in the statement of Lemma~\ref{flip-lemma} cannot be dropped; the figure to the right shows two crossing segments $pq$ and $rs$ for which $(|pq|+|rs|)-(|ps|+|rq|)$ tends to zero as $\angle roq$ tends to $\pi$.}

\begin{proof}[Proof of Lemma~\ref{flip-lemma1}]
	We recall the simple fact that the largest side of every triangle always faces the largest angle of the triangle. 
	Since $\angle cob\leqslant \pi/3$, we have that $\angle cbo\geqslant \pi/3$ or $\angle bco\geqslant \pi/3$. Without loss of generality assume that $\angle bco \geqslant \pi/3$, and thus, $\angle bco \geqslant\angle cob$; \changed{see Figure~\ref{dense-fig}(b)}. This implies that $|ob|\geqslant|cb|$. By a similar reasoning, we get that $|oa|\geqslant |ad|$ or $|od|\geqslant |ad|$. If $|oa|\geqslant|ad|$, then
	$$|ab|+|cd|-(|ad|+|cb|)= (|oa|+|ob|)+|cd|-(|ad|+|cb|)\geqslant|cd|\geqslant \cpd'',$$
	and if $|od|\geqslant|ad|$, then		
	$$|ab|+|cd|-(|ad|+|cb|)= (|oa|+|ob|)+(|oc|+|od|)-(|ad|+|cb|)\geqslant|oa|+|oc|\geqslant|ac|\geqslant \cpd''.\qedhere$$
\end{proof}

\begin{proof}[Proof of Lemma~\ref{flip-lemma2}]
	Refer to Figure~\ref{dense-fig}(c) for an illustration of the proof. Let $o$ be the intersection point of $ab$ and $cd$. Let $o'$ be the intersection point between $cb$ and the line \changed{through $o$} that is perpendicular to $cb$. Without loss of generality assume that $ob$ is longer than $oc$, i.e., $|ob|\geqslant |oc|$. Then $\angle ocb \geqslant \angle obc$, and thus, $\angle ocb\geqslant \pi/4$. Since $\angle oo'c=\pi/2$ and $\angle oco'=\angle ocb \geqslant \pi/4$, we get that $\angle coo'$ is the smallest angle in the triangle $\bigtriangleup oco'$, and thus, $o'c$ is its smallest side. \changed{Based on this and the fact that $\bigtriangleup coo'$ is a right triangle, we have
	
	$$|oc|^2=\sqrt{|oo'|^2+|o'c|^2}\geqslant \sqrt{2|o'c|^2},$$
	which implies that
	$|o'c|\leqslant |oc|/\sqrt{2}$.}
	
	Let $o''$ be the intersection point between $ad$ and the line \changed{through $o$} that is perpendicular to $ad$. We consider two cases depending on which of $oa$ and $od$ is longer.
	\begin{itemize}
		\item $|oa|> |od|$:		
		By a similar reasoning as for $ob$ and $oc$ we get that $|o''d|\leqslant |od|/\sqrt{2}$.
		Observe that $|ob|>|o'b|$ and $|oa|>|o''a|$. By combining these inequalities we get
		\begin{align}
		\notag	 (|ab|+|cd|)-(|ad|+|cb|)&=(|oa|+|ob|)+(|oc|+|od|)-(|o''a|+|o''d|)-(|o'c|+|o'b|)\\
		\notag							& >|oc|+|od|-|o''d|-|o'c|\geqslant |oc|+|od|-\frac{|od|}{\sqrt{2}}-\frac{|oc|}{\sqrt{2}}\\
		\notag							& = \left(1-\frac{1}{\sqrt{2}}\right)(|oc|+|od|)= \notag		\frac{2-\sqrt{2}}{2}|cd|\geqslant \frac{2-\sqrt{2}}{2}\cpd''.
		\end{align}
		\item $|oa|\leqslant |od|$:		
		Again, by a similar reasoning as for $ob$ and $oc$ we get that $|o''a|\leqslant |oa|/\sqrt{2}$. Also, by a similar reasoning as in the previous case we get		
		\begin{align}
		\notag	 (|ab|+|cd|)-(|ad|+|cb|)&\geqslant |oc|+|oa|-\frac{|oa|}{\sqrt{2}}-\frac{|oc|}{\sqrt{2}}\\
		\notag							& =	\frac{2-\sqrt{2}}{2}|ca|\geqslant \frac{2-\sqrt{2}}{2}\cpd''.
		\end{align}
	\end{itemize}
\end{proof}

\section{Points in Convex Position}
\label{convex-section}
In this section we study the problem of uncrossing perfect matchings and spanning trees on points in convex position.
For perfect matchings, Bonnet and Miltzow \cite{Bonnet2016} exhibited two perfect matchings $M_1$ and $M_2$ on $n$ points in the plane in convex position such that $f(M_1)\geqslant n/2-1$ and $F(M_2)\geqslant {n/2\choose 2}$. The following theorem provides matching upper bounds for $f(\cdot)$ and $F(\cdot)$.

\begin{theorem}
	\label{matching-convex-thr}
	For every perfect matching $M$ on a set of $n$ points in the plane in convex position we have $f(M)\leqslant \frac{n}{2}-1$ and $F(M)\leqslant {n/2\choose 2}$.
\end{theorem}

\begin{proof}
	The matching $M$ contains $n/2$ edges. First we prove that $F(M)\leqslant {n/2\choose 2}$. Notice that the number of crossings between the edges of $M$ is at most $n/2\choose 2$. We show that any flip reduces this number by at least one, and thus, our claim follows. Take any pair $ab$ and $cd$ of crossing edges of $M$. Flip this crossing, and let $ac$ and $bd$ be the new edges, after a suitable relabeling. After this flip operation, the crossing between $ab$ and $cd$ disappears. Moreover, any edge of $M$ that crosses $ac$ (or $bd$) used to cross $ab$ or $cd$, and any edge of $M$ that crosses both $ac$ and $bd$ used to cross both $ab$ and $cd$. Therefore, the total number of crossings reduces by at least one, and thus, our claim follows. 
	
	Now, we prove, by induction on $n$, that $f(M)\leqslant n/2-1$. If $n=2$, then $M$ has only one edge, and thus, $f(M)=0$. Assume that $n\geqslant 4$. First, we show how to transform $M$, by at most one flip, to a perfect matching $M'$ containing a boundary edge, i.e., an edge of the boundary of the convex hull. Let $p_1,\dots, p_n$ be the points in clockwise order. Let $p_ip_j$ be an edge of $M$ with minimum depth $m$. If $m=0$, then $M'=M$ is a matching in which $p_ip_j$ is a boundary edge. Suppose that $m\geqslant 1$. Without loss of generality assume that $i=1$ and $j=m+2$. Let $p_k$ be the point that is matched to $p_2$ by $M$. Because of the minimality of $m$, the edge $p_2p_k$ crosses $p_1p_{m+2}$. By flipping $p_2p_k$ and $p_1p_{m+2}$ to $p_1p_2$ and $p_{m+2}p_k$ we obtain $M'$ in which $p_1p_2$ is a boundary edge. Let $M''$ be the matching on $n-2$ points obtaining from $M'$ by removing a boundary edge. By the induction hypothesis, it holds that \vspace{-2pt}
	\[f(M)\leqslant 1 + f(M'')\leqslant 1+\left(\frac{n-2}{2}-1\right) = \frac{n}{2}-1.\qedhere\]
\end{proof}

In the rest of this section we study spanning trees. The argument of \cite{Leeuwen1981} for Hamiltonian cycles also extends to spanning trees, that is, if $T_1$ is a spanning tree on $n$ points in the plane, then $F(T_1)=O(n^3)$. Also, by an argument similar to the one in the proof of Theorem~\ref{matching-convex-thr}, it can easily be shown that for every spanning tree $T$ on $n$ points in the plane in convex position we have that $F(T)=O(n^2)$. In this section we prove that $f(T)=O(n\log n)$. Recall that a boundary edge is an edge of the boundary of the convex hull.

\begin{lemma}
	\label{tree-boundary-lemma}
	Any spanning tree on a point set in convex position can be transformed, by at most two flips, into a spanning tree containing a boundary edge.
\end{lemma}
\begin{proof}
	Let $T$ be a spanning tree on $n$ points in the plane in convex position, and let $p_1,\dots, p_n$ be the points in clockwise order. Let $p_ip_j$ be an edge of $T$ with minimum depth $m$ (recall the definition of depth from Section~\ref{preliminaries-section}). If $m=0$, then $p_ip_j$ is a boundary edge. Suppose that $m\geqslant 1$. Without loss of generality assume that $i=1$ and $j=m+2$. Because of the minimality of $m$, all edges of $T$ that are incident on $p_2,\dots,p_{m+1}$ cross $p_1p_{m+2}$. We consider two cases with $m=1$ and $m>1$.
	\begin{itemize}
		\item $m=1$. In this case $p_{m+1}=p_2$ and $p_{m+2}=p_3$. Let $\delta$ be the path between $p_2$ to $p_3$ in $T$, and let $p_k$ be the vertex that is adjacent to $p_2$ in $\delta$. If $\delta$ contains $p_1$, then we flip $p_1p_3$ and $p_2p_k$ to $p_1p_2$ and $p_3p_k$; this gives a spanning tree in which $p_1p_2$ is a boundary edge. If $\delta$ does not contain $p_1$, then we flip $p_1p_3$ and $p_2p_k$ to $p_2p_3$ and $p_1p_k$; this gives a spanning tree in which $p_2p_3$ is a boundary edge. 
		\item $m>1$. Let $\delta$ be the path between $p_2$ to $p_{m+2}$ in $T$, and let $p_k$ be the vertex that is adjacent to $p_2$ in $\delta$. If $\delta$ contains $p_1$, then we flip $p_1p_{m+2}$ and $p_2p_k$ to $p_1p_2$ and $p_{m+2}p_k$; this gives the a spanning tree in which $p_1p_2$ is a boundary edge. Assume that $\delta$ does not contain $p_1$. Let $\delta'$ be the path between $p_{m+1}$ \changed{and} $p_1$ in $T$, and let $p_{k'}$ be the vertex that is adjacent to $p_{m+1}$ in $\delta'$; it may be that $k'=k$. If $\delta'$ contains $p_{m+2}$, then we flip $p_1p_{m+2}$ and $p_{m+1}p_{k'}$ to $p_{m+1}p_{m+2}$ and $p_1p_{k'}$; this gives a spanning tree in which $p_{m+1}p_{m+2}$ is a boundary edge. Assume that $\delta'$ does not contain $p_{m+2}$. See Figure~\ref{tree-fig}(a). In this case we have that $k'\neq k$, because otherwise $T$ would have a cycle. First we flip $p_1p_{m+2}$ and $p_{m+1}p_{k'}$ to $p_{1}p_{m+1}$ and $p_{m+2}p_{k'}$, then we flip $p_1p_{m+1}$ and $p_2p_k$ to $p_1p_2$ and $p_kp_{m+1}$; this gives a spanning tree in which $p_1p_2$ is a boundary edge. \qedhere
	\end{itemize} 
\end{proof}

\begin{figure}[htb]
	\centering
	\setlength{\tabcolsep}{0in}
	$\begin{tabular}{cc}
	\multicolumn{1}{m{.55\columnwidth}}{\centering\includegraphics[width=.51\columnwidth]{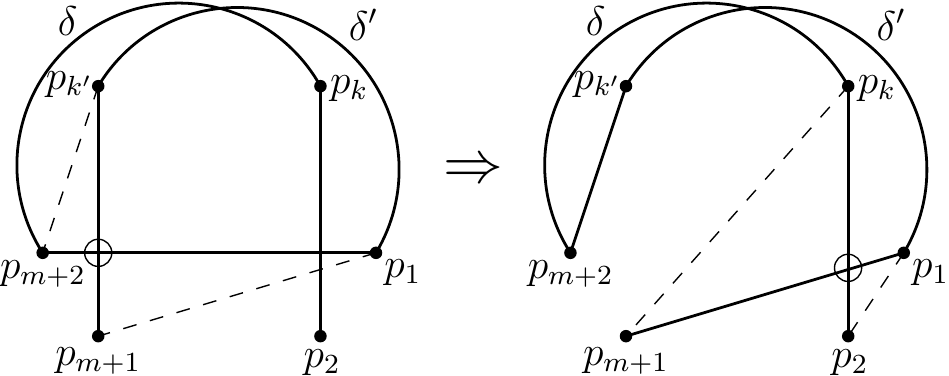}}
	&\multicolumn{1}{m{.45\columnwidth}}{\centering\includegraphics[width=.35\columnwidth]{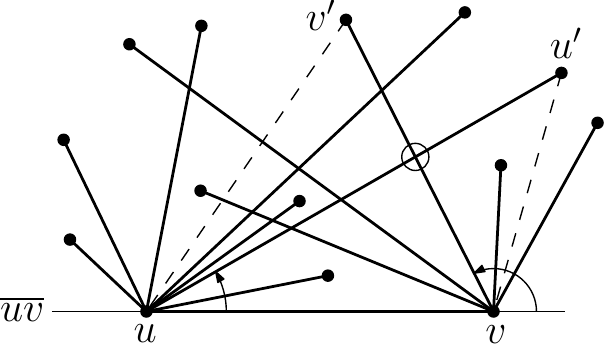}}
	\\
	(a) & (b)
	\end{tabular}$
	\caption{(a) Illustration of the proof of Lemma~\ref{tree-boundary-lemma}: Flipping $p_1p_{m+2}$ and $p_{m+1}p_{k'}$ to $p_{1}p_{m+1}$ and $p_{m+2}p_{k'}$, and then flipping $p_1p_{m+1}$ and $p_2p_k$ to $p_1p_2$ and $p_kp_{m+1}$. (b) Illustration of the proof of Lemma~\ref{tree-lemma}: $vv'$ is the first counterclockwise edge incident on $v$ that is crossed by some edges incident on $u$, and $uu'$ is the first counterclockwise edge incident on $u$ that crosses $vv'$.}
	\label{tree-fig}
\end{figure}

For the following lemma we do not need the vertices to be in convex position.

\begin{lemma}
	\label{tree-lemma}
	Let $T$ be a spanning tree containing an edge $uv$ such that every other edge is incident on either $u$ or $v$. Then $f(T)\leqslant \min(\deg{(u)},\deg{(v)})-1$, and this bound is tight.
\end{lemma}
\begin{proof}
	After a suitable rotation and/or a horizontal reflection and/or relabeling assume that $\linepq{u}{v}$ is horizontal, $u$ is to the left of $v$, and that $\deg{(v)}\leqslant \deg{(u)}$. The edges that are incident on points above $\linepq{u}{v}$ do not cross the edges incident on points below $\linepq{u}{v}$. Thus, the crossings above $\linepq{u}{v}$ can be handled independently of the ones below $\linepq{u}{v}$. Because of symmetry, we describe how to handle the crossings above $\linepq{u}{v}$. See Figure~\ref{tree-fig}(b). We show how to increase, by one flip, the number of free edges that are incident on $v$. By repeating this process, our claim follows. To that end, let $v'$ be the first vertex, in counterclockwise order, that is adjacent to $v$, and such that $vv'$ is crossed by at least one edge incident on $u$. Let $u'$ be the first vertex, in counterclockwise order, that is adjacent to $u$, and such that $uu'$ crosses $vv'$; see Figure~\ref{tree-fig}(b). Flip this crossing to obtain new edges $vu'$ and $uv'$. The edge $vu'$ is free, because otherwise $uu'$ cannot be the first counterclockwise edge that crosses $vv'$. Moreover, any edge that is crossed by $uv'$ used to be crossed by $uu'$. Thus, the number of free edges that are incident on $v$ increases by at least one. \changed{Notice that a flip operation does not change the degrees of vertices.} Therefore, by repeating this process, after at most $\deg{(v)}-1$ iterations, all incident edges on $v$ become free (notice that the edge $uv$ is already free); this transforms $T$ to a plane spanning tree. This proves the first statement of the lemma.
	
	Recall that the statement of this lemma is not restricted to points in convex position, and thus, the vertices of our tight example do not need to be in convex position. To verify the tightness of the bound, consider a tree in which every edge incident on $v$ (except $uv$) is crossed by exactly one of the edges incident on $u$, and every edge incident on $u$ crosses at most one of the edges incident on $v$. This tree needs exactly $\deg{(v)}-1$ flips to be transformed to a plane tree. 
\end{proof}

\begin{theorem}
	\label{tree-thr}
	For every spanning tree $T$ on $n$ points in the plane in convex position we have that $f(T)=O(n\log n)$.
\end{theorem}
\begin{proof}
	We present a recursive algorithm that uncrosses $T$ by $O(n\log n)$ flips. As for the base case, if $n\leqslant 3$, then $T$ is plane, and thus, no flip is needed. Assume that $n\geqslant 4$. By Lemma~\ref{tree-boundary-lemma}, by at most two flips, we can transform $T$ to a tree $T'$ containing a boundary edge $uv$. Contract the edge $uv$ and denote the resulting tree with $n-1$ vertices by $T''$; this can be done by removing the vertex $u$ together with its incident edges, and then connecting its neighbors, by straight-line edges, to $v$. We call every such new edge a {\em $u$-edge}. Recursively uncross $T''$ with $f(T'')$ flips. During the uncrossing process of $T''$, whenever we flip/remove a $u$-edge, we call the new edge that gets connected to $v$ a $u$-edge. After uncrossing $T''$ we return the vertex $u$ back and connect it to $v$. Then we remove every $u$-edge $vv'$, which is incident on $v$, and connect $v'$ to $u$. In the resulting tree, every crossing is between an edge that is incident on $u$ and an edge that is incident on $v$. Thus, after at most $2+f(T'')$ flips, $T$ can be transformed into a tree in which any two crossing edges are incident on $u$ and $v$. Then by Lemma~\ref{tree-lemma}, we can obtain a plane tree by performing at most $\min(\deg(u), \deg(v))-1$ more flips. Notice that the flip operation does not change the degree of vertices, and thus, every vertex in the resulting tree has the same degree as in $T$. Therefore, we have that
	\begin{align*}
	f(T)&\leqslant 2+f(T'')+\min(\deg{(u)},\deg{(v)})-1\\
	&= 1+\min(\deg{(u)},\deg{(v)})+f(T'').
	\end{align*} 
	
	It remains to show that $f(T)=O(n\log n)$. To that end, we interpret the above recursion by a union-find data structure with the linked-list representation and the weighted-union heuristic \cite[Chapter 21]{Cormen2001}.
	The number of flips in the above recursion can be interpreted as the total time for union operations as follows: each time that we contract an edge $uv$ and recurse on a smaller tree we perform at most $1+\min(\deg{(u)},\deg{(v)})$ flips. Consider every vertex $x$ of $T$ as a set with $\deg{(x)}$ elements. Also, assume that all the elements of these sets are pairwise distinct. Thus, we have $n$ disjoint sets of total size $2(n-1)$; this is coming from the fact that $T$ has $n-1$ edges and its total vertex degree is $2(n-1)$. The contraction of an edge $uv$ can be interpreted as a union operation of the sets $u$ and $v$ whose cost (number of flips) is at most $1+\min(|u|,|v|)$, where $|x|$ denotes the size of the set $x$. From the union-find data structure we have that the cost of a sequence of $s$ operations on $m$ elements is $O(s+m\log m)$. In our case, the number $m$ of elements is $2(n-1)$, and the number $s$ of union operations (edge contractions) is $n-3$ (no contraction is needed when we hit the base case). Thus, it follows that the total cost (the total number of flips) is $O(n\log n)$.  
\end{proof}

\section{Bichromatic Matchings}
\label{bichromatic-section}

In this section we study the problem of uncrossing perfect bichromatic matchings for points in convex position and for semi-collinear points. Let $n\geqslant 2$ be an even integer, and let $P$ be a set of $n$ points in the plane, $n/2$ of which are colored red and $n/2$ are colored blue. If $P$ is in general position, then for any bichromatic matching $M$ on $P$, the best known upper bound for both $f(M)$ and $F(M)$ is the $O(n^3)$ bound that has been proved in \cite{Bonnet2016, Leeuwen1981}. If $P$ is in convex position, the $n/2-1$ and $n/2\choose 2$ lower bounds that are shown in \cite{Bonnet2016} for $f(\cdot)$ and $F(\cdot)$, respectively, in the monochromatic setting, also hold in the bichromatic setting. Theorem~\ref{matching-convex-thr} implies that the $n/2\choose 2$ bound for $F(\cdot)$ is tight. The following theorem gives an upper bound on $f(\cdot)$ for points in convex position.

\begin{figure}[htb]
	\centering
	\includegraphics[width=.67\columnwidth]{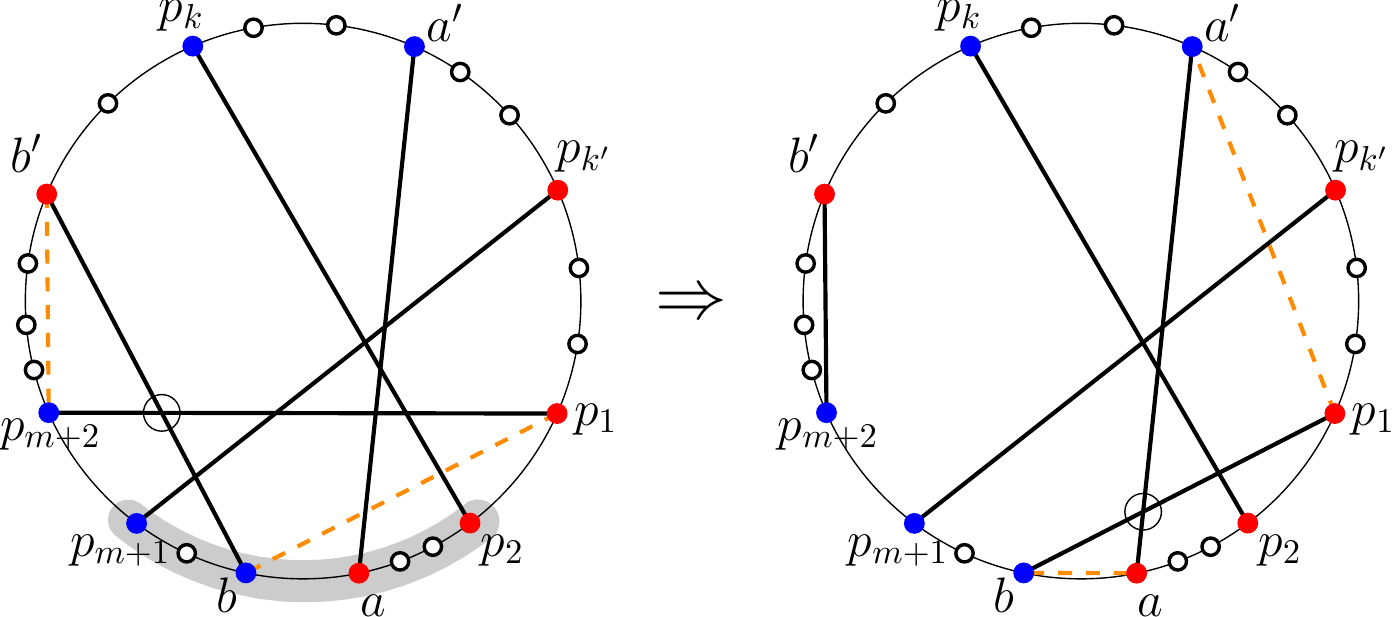}
	\caption{Illustration of the proof of Theorem~\ref{bimatching-convex-thr}. Flipping $bb'$ and $p_1p_{m+2}$ to $b'p_{m+2}$ and $bp'$, and then flipping $bp_1$ and $aa'$ to $p_1a'$ and $ab$.}
	\label{bimatching-convex-fig}
\end{figure}

\begin{theorem}
	\label{bimatching-convex-thr}
	For every perfect bichromatic matching $M$ on $n$ points in the plane in convex position we have that $f(M)\leqslant n-2$.
\end{theorem}

\begin{proof}
	Our proof is by induction on $n$. 
	If $n=2$, then $f(M)=0$. Assume that $n\geqslant 4$. First we show how to transform $M$, by at most two flips, to a perfect bichromatic matching $M'$ containing a boundary edge. Let $p_1,\dots, p_n$ be the points in clockwise order. Let $p_ip_j$ be an edge of $M$ with minimum depth $m$. If $m=0$, then $M'=M$ is a matching in which $p_ip_j$ is a boundary edge. Suppose that $m\geqslant 1$. Without loss of generality assume that $i=1$, $j=m+2$, $p_1$ is red, and $p_{m+2}$ is blue as in Figure~\ref{bimatching-convex-fig}. Let $p_k$ and $p_{k'}$ be the points that are matched to $p_2$ and $p_{m+1}$, respectively; it may be that $m+1=2$ and $k'=k$. Because of the minimality of $m$, all edges that are incident on points $p_2,\dots, p_{m+1}$ cross $p_1p_{m+2}$. If $p_2$ is blue, then by flipping $p_1p_{m+2}$ and $p_2p_k$ to $p_1p_2$ and $p_{m+2}p_k$ we obtain $M'$ in which $p_1p_2$ is a boundary edge. Assume that $p_2$ is red. If $p_{m+1}$ is red, then by flipping $p_1p_{m+2}$ and $p_{m+1}p_{k'}$ to $p_{m+1}p_{m+2}$ and $p_1p_{k'}$ we obtain $M'$ in which $p_{m+1}p_{m+2}$ is a boundary edge. Assume that $p_{m+1}$ is blue. See Figure~\ref{bimatching-convex-fig}. \changed{It follows that} $p_2$ and $p_{m+1}$ have different colors, and thus, $m+1\neq 2$ and $k'\neq k$.
	
	For an illustration of the rest of the proof, follow Figure~\ref{bimatching-convex-fig}. The sequence $p_2,\dots, p_{m+1}$ starts with a red point and ends with a blue point. Thus, in this sequence there are two points of distinct colors, say $a$ and $b$, that are consecutive. Let $b$ be the first blue point after $p_1$. Let $a'$ and $b'$ be the two points that are matched to $a$ and $b$ respectively.  By flipping $bb'$ and $p_1p_{m+2}$ to $b'p_{m+2}$ and $bp_1$, and then flipping $bp_1$ and $aa'$ to $p_1a'$ and $ab$ we obtain $M'$ in which $ab$ is a boundary edge.
	
	Let $M''$ be the bichromatic matching on $n-2$ points obtaining from $M'$ by removing a boundary edge.  
	By the induction hypothesis, it holds that \vspace{-5pt}
	\[f(M)\leqslant 2 + f(M'')\leqslant 2+\left((n-2)-2\right) = n-2.\qedhere\]
\end{proof}

In the rest of this section we study the case where $P$ is semi-collinear, i.e., its blue points are on a straight line and its red points are in general position. Semi-collinear points have been studied in may problems related to plane matchings (see e.g., \cite{Aloupis2013, Biniaz2014a, Carlsson2015}). We prove that for every perfect bichromatic matching $M$ on $P$, it holds that $f(M)=O(n^2)$. 
Before we prove this upper bound, observe that similar to the general position setting, in the semi-collinear setting the total number of crossings might increase after a flip. Also, it is possible that a crossing, that has disappeared after a flip, reappears after some more flips (see the crossing between $br$ and $b'r'$ in Figure~\ref{reappear-fig}). The $O(n^2)$ upper bound given in \cite{Bonnet2016} for $f(\cdot)$ on uncolored points, which is obtained by connecting the two leftmost points of a crossing, does not apply to our semi-collinear bichromatic setting, because in this setting the two leftmost points might have the same color. These observations \changed{suggest} that there is no straightforward way of getting a good upper bound.

\begin{figure}[htb]
	\centering
	\includegraphics[width=.98\columnwidth]{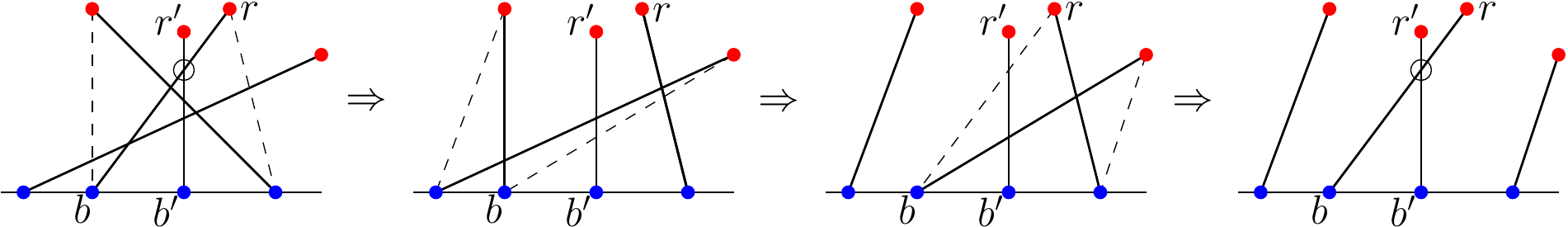}
	\caption{Reappearance of the crossing between $br$ and $b'r'$.}
	\label{reappear-fig}
\end{figure}

Let $\ell$ be the line that contains all the blue points of $P$. By a suitable rotation we assume that $\ell$ is horizontal.
For every perfect bichromatic matching $M$ on $P$, the edges of $M$, that are above $\ell$, do not cross the ones that are below $\ell$. Thus, we can handle these two sets of edges independently of each other. Therefore, in the rest of this section we assume that the red points of $P$ lie above $\ell$. Recall that $P$ contains $n/2$ blue points and $n/2$ red points.

\begin{lemma}
	\label{rightmost-crossing-lemma}
	Let $M$ be a perfect bichromatic matching on $P$ in which the rightmost blue point $b$ is matched to the topmost red point $r$. If $M\setminus\{br\}$ is plane, then $f(M)\leqslant \frac{n}{2}-1$, and this bound is tight.
\end{lemma}
\begin{proof}
	See Figure~\ref{rightmost-crossing-fig}(a) for an illustration of the statement of this lemma; notice that if we remove $br$ from $M$, then we get a plane matching. Our proof is by induction on $n$. If $n=2$, then $M$ has one edge which is plane, and thus, $f(M)=0$. Assume that $n\geqslant 4$. If $br$ does not intersect any other edge, then $M$ is plane and $f(M)=0$. Suppose that $br$ intersects some edges of $M\setminus\{br\}$, and let $R'$ be the set of the red endpoints of those edges; see Figure~\ref{rightmost-crossing-fig}(a). Let $r'$ be the first red point in the counterclockwise order of the red points around $b$; observe that $r'$ belongs to $R'$. Let $b'$ be the blue point that is matched to $r'$. Flip $br$ and $b'r'$ to $br'$ and $b'r$ as in Figure~\ref{rightmost-crossing-fig}(b), and let $M'$ be the resulting matching. The edge $br'$ does not cross any other edge of $M'$, because of our choice of $r'$, but the edge $b'r$ may cross some edges of $M'$. Let $M''$ be the subset of edges of $M'$ that are to the left of $\linepq{b'}{r'}$; see Figure~\ref{rightmost-crossing-fig}(b). Notice that $br'\notin M''$, and thus $M''$ is a matching on at most $n-2$ points. 
	
	\begin{figure}[htb]
		\centering
		\setlength{\tabcolsep}{0in}
		$\begin{tabular}{ccc}
		\multicolumn{1}{m{.35\columnwidth}}{\centering\includegraphics[width=.33\columnwidth]{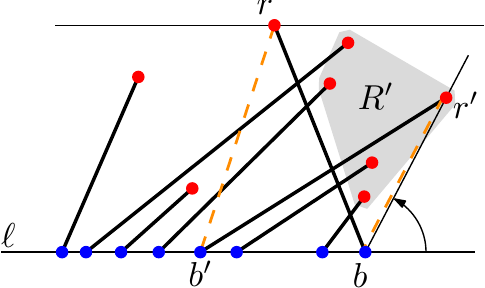}}
		&\multicolumn{1}{m{.37\columnwidth}}{\centering\includegraphics[width=.34\columnwidth]{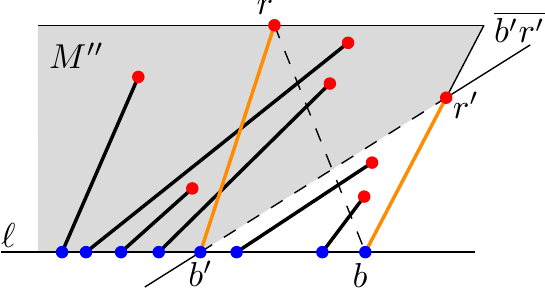}}
		&\multicolumn{1}{m{.28\columnwidth}}{\centering\includegraphics[width=.28\columnwidth]{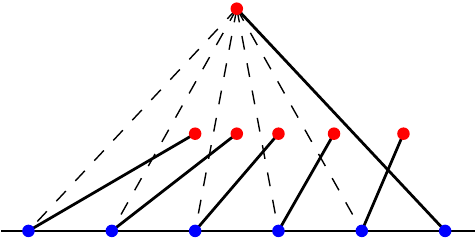}}
		\\
		(a) & (b)&(c)
		\end{tabular}$
		\caption{Illustration of the proof of Lemma~\ref{rightmost-crossing-lemma}.}
		\label{rightmost-crossing-fig}
	\end{figure}
	
	Because of the planarity of $M\setminus\{br\}$ and since $r$ is the topmost red point, we have that $M'\setminus M''$ is plane. Moreover, $M''$ and $M'\setminus M''$ are separated by $\linepq{b'}{r'}$. 
	Observe that $b'$ is the rightmost blue point in $M''$ that is matched to the topmost red point $r$, moreover, $M''\setminus\{b'r\}$ is plane. Therefore, we can repeat the above process on $M''$, which is a smaller instance of the initial problem. By the induction hypothesis, it holds that $$f(M)= 1 + f(M'')\leqslant 1+\left(\frac{n-2}{2}-1\right) = \frac{n}{2}-1.$$

	To verify the tightness, Figure~\ref{rightmost-crossing-fig}(c) shows a matching example for which we need exactly $n/2{-}1$ flips to transform it to a plane matching. Each time there exists exactly one crossing, and after flipping that crossing, only one other crossing appears (except for the last flip). 
\end{proof}

\begin{figure}[htb]
	\centering
	\setlength{\tabcolsep}{0in}
	$\begin{tabular}{cc}
	\multicolumn{1}{m{.5\columnwidth}}{\centering\includegraphics[width=.49\columnwidth]{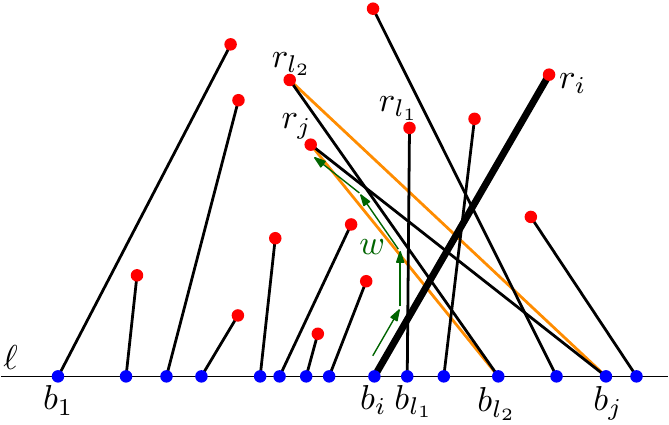}}
	&\multicolumn{1}{m{.5\columnwidth}}{\centering\includegraphics[width=.49\columnwidth]{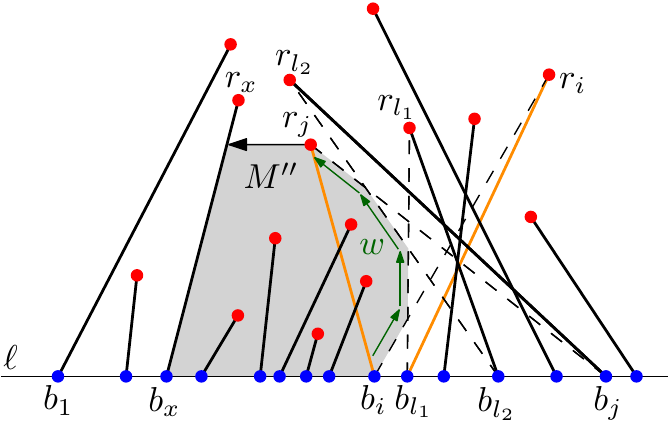}}
	\end{tabular}$
	\vspace{-5pt}
	\caption{Illustration of the proof of Theorem~\ref{semi-collinear-thr}. \changed{From the left configuration, we will get to the right configuration by flipping the following crossing pairs: $(r_jb_j,r_{l_2}b_{l_2})$, $(r_jb_{l_2},r_{l_1}b_{l_2})$, and $(r_jb_{l_1},r_ib_i)$.}}
	\label{rightmost-crossing-fig2}
\end{figure}
\vspace{-10pt}

\begin{theorem}
	\label{semi-collinear-thr}
	For every perfect bichromatic matching $M$ on $P$ we have \changed{$f(M)\leqslant \frac{n^2}{4}-\frac{n}{2}$}.
\end{theorem}
\begin{proof}
	We present an iterative algorithm that uncrosses $M$ by $O(n^2)$ flips. Let $b_1,\dots,b_{n/2}$ be the blue points from left to right. By a suitable relabeling assume that $M=\{b_1r_1,\dots,\allowbreak b_{n/2}r_{n/2}\}$. To simplify the description of the proof, we add, to $M$, a dummy edge $b_0r_0$ such that $b_0$ is a blue point on $\ell$ that is to the left of all the blue points, $r_0$ is a red point that is higher than all the red points, and all points of $P$ are to the right of $\linepq{b_0}{r_0}$. 
	
	We describe one iteration of our algorithm. If $M$ is plane, then the algorithm terminates. Assume that $M$ is not plane. Let $i\in\{1,\dots,n/2\}$ be the smallest index such that $b_ir_i$ intersects some edges of $M$; see Figure~\ref{rightmost-crossing-fig2}-left. To simplify the rest of our description, we refer to the current iteration as iteration $i$. Notice that the blue endpoint of every non-free edge is strictly to the right of $b_{i-1}$. \changed{We perform the following walk along the edges of $M$ and stop as soon as we meet a red point. Starting from $b_i$, we walk along $b_ir_i$ until we see the first edge, say $b_{l_1}r_{l_1}$, that crosses $b_ir_i$. Then we turn left on $b_{l_1}r_{l_1}$ and keep walking until we see another edge, say  $b_{l_2}r_{l_2}$, that crosses $b_{l_1}r_{l_1}$. Then we turn left on $b_{l_2}r_{l_2}$ and continue this process until we meet a red point for the first time. Let $r_j$ denote this red point and let $b_j$ be the blue point that is matched to $r_j$. Let $\delta=b_ir_i,b_{l_1}r_{l_1}, b_{l_2}r_{l_2},\dots,b_{l_k}r_{l_k},b_jr_j$ be the sequence of edges visited along this walk, and $w$ denote the convex polygonal path that we have traversed; see Figure~\ref{rightmost-crossing-fig2}-left. 
	
	We are going to use $\delta$ to connect $r_j$ with $b_i$ by at most $j-i$ flips in such a way that none of the new edges intersects $w$ (except $r_jb_i$ which will connect the two endpoints of $w$). Observe that the blue points $b_i,b_{l_1}, b_{l_2},\dots,b_{l_k},b_j$ are ordered from left to right. We repeatedly flip the edge incident to $r_j$ until we connect $r_j$ with $b_i$. To that end, we first flip $r_jb_j$ and $r_{l_k}b_{l_k}$ (which are crossing) to $r_jb_{l_k}$ and $r_{l_k}b_j$. The edge $r_{l_k}b_j$ is to the right of $w$. If $b_{l_k}\neq b_i$, then the edge $r_jb_{l_k}$ intersects some other edge $r_{l_{k'}}b_{l_{k'}}$ such that $r_{l_{k'}}b_{l_{k'}} \in \delta$ and $k'<k$. In this case we flip $r_jb_{l_k}$ and $r_{l_{k'}}b_{l_{k'}}$ to $r_jb_{l_{k'}}$ $r_{l_{k'}}b_{l_{k}}$. The edge $r_{l_{k'}}b_{l_{k}}$ is to the right of $w$; in Figure~\ref{rightmost-crossing-fig2}-left we have $k=2$ and $k'=1$. Then we repeat this process with $r_jb_{l_{k'}}$ until $r_j$ gets connected to $b_i$. It turns out that after every flip, $r_j$ connects to a blue point that is closer to $b_i$, and the other edge is to the right of $w$. Therefore after at most $j-i$ flips we get the edge $r_jb_i$. Since $j\leqslant n/2$, the total number of flips in the above process is at most $n/2-i$.}
	
	Let $M'$ be the matching obtained after the above process. See Figure~\ref{rightmost-crossing-fig2}-right. Shoot a horizontal ray from $r_j$, to the left, and stop as soon as it hits an edge $b_xr_x$ in $M'$ \changed{with $x\in\{0,\dots,i-1\}$ (the dummy edge $b_0r_0$ ensures that such an index $x$ exists)}. Let $M''$ be the subset of the edges of $M'$ that are incident on $b_{x+1},\dots,b_i$, that is, $M''=\{b_{x+1}r_{x+1},\dots, b_{i-1}r_{i-1}, b_ir_j\}$. \changed{By our choices of $r_j$ and the flips in the above process,} the edges of $M''$ are in a convex region whose interior is disjoint from the edges of $M'\setminus M''$; this convex region is bounded by $\ell$, $b_xr_x$, $w$, and the ray from $r_j$, as depicted in Figure~\ref{rightmost-crossing-fig2}-right. The matching $M''$ has $i-x$ edges. Observe that, in $M''$, we have that $b_i$ is the rightmost blue point that is matched to the topmost red point $r_j$, and $M''\setminus\{b_ir_j\}$ is plane. Thus, by Lemma~\ref{rightmost-crossing-lemma} we can uncross $M''$ by at most $i-x-1$ flips. To this end, we have transformed $M$ to a matching in which the edges that are incident on $b_1,\dots,b_i$ are free. The total number of flips performed in iteration $i$ is at most \changed{$(n/2-i)+(i-x-1)=n/2-x-1\leqslant n/2-1$.}
	
	In the next iteration, the smallest index $i'$, for which $b_{i'}r_{i'}$ is not free, is larger than $i$. Thus, this smallest index moves at least one step to the right after each iteration. This means that the number of free edges, that are connected to the blue points of lower indices, increases. Therefore, after at most $n/2$ iterations our algorithm terminates. The total number of flips is 
	\[\changed{f(M)\leqslant \sum_{i=1}^{n/2}\left(\frac{n}{2}-1\right)=\frac{n^2}{4}-\frac{n}{2}.}\qedhere\]
\end{proof}

\section{Conclusions}
We investigated the number of flips that are necessary and sufficient to reach a non-crossing perfect matching on $n$ points in the plane. It is known that the minimum and the maximum lengths of a flip sequence are $O(n^2)$ and $O(n^3)$, respectively. We proved, with a new approach, that the minimum length of a flip sequence is $O(n\Delta)$ where $\Delta$ is the spread of the points set; this improves the previous bound \changed{if the point set has sublinear spread.} A natural open problem is to improve any of these bounds. Another open problem is to improve our $O(n\log n)$ upper bound on the number of sufficient flips to reach a plane spanning tree on points in convex position, or to show that this bound is tight. One potential way to do this, is that in Theorem~\ref{tree-thr}, we get a boundary edge $uv$ such that one of $u$ or $v$ has a constant degree.

It is worth mentioning that the number of flips, in a flip sequence, is highly dependent on the {\em order} in which we choose crossings to flip, and the {\em type} of a flip that we perform (among the two possible types). This dependency can be used to improve the bounds on the minimum number of flips. In Theorems~\ref{matching-convex-thr}, \ref{tree-thr}, \ref{bimatching-convex-thr}, and \ref{semi-collinear-thr} we used the order and proved some upper bounds, while in Theorem~\ref{dense-thr} we used the flip type. One may think of using the order and the flip type together to improve the current bounds.
Notice that for bichromatic matchings, spanning trees, and Hamiltonian cycles only one type of flip is possible, and thus, only the order can be used for further improvements. Also, notice that none of the order and the flip type can be used to improve the bounds on the maximum number of flips, because, in this case, an adversary chooses the order and the type.

\changed{
\section*{Acknowledgement}
	We would like to thank anonymous referees for their comments which improved the readability of the paper, and for pointing out an issue with the original proof of Theorem~\ref{semi-collinear-thr}.}

\bibliographystyle{abbrv}
\bibliography{Flip-Distance.bib}
\end{document}